\documentclass{llncs}
\usepackage{bris}

\title{On Instruction Sets for Boolean Registers \\ in Program Algebra}
\author{J.A. Bergstra \and C.A. Middelburg}
\institute{Informatics Institute, Faculty of Science, University of
           Amsterdam, \\
           Science Park~904, 1098~XH Amsterdam, the Netherlands \\
           \email{J.A.Bergstra@uva.nl,C.A.Middelburg@uva.nl}}

\begin{document}
\maketitle

\begin{abstract}
In previous work carried out in the setting of program algebra, 
including work in the area of instruction sequence size complexity, we 
chose instruction sets for Boolean registers that contain only 
instructions of a few of the possible kinds.
In the current paper, we study instruction sequence size bounded 
functional completeness of all possible instruction sets for Boolean 
registers.
We expect that the results of this study will turn out to be useful to 
adequately assess results of work that is concerned with lower bounds of 
instruction sequence size complexity.
\begin{classcode}
F.1.1, F.1.3.
\end{classcode}\end{abstract}

\section{Introduction}
\label{sect-intro}

In~\cite{BM13a}, we presented an approach to computational complexity in 
which algorithmic problems are viewed as Boolean function families that 
consist of one $n$-ary Boolean function for each natural number $n$ and 
the complexity of such problems is assessed in terms of the length of 
finite single-pass instruction sequences acting on Boolean registers 
that compute the members of these families.
The instruction sequences concerned contain only instructions to set 
and get the content of Boolean registers, forward jump instructions, 
and a termination instruction.
Moreover, each Boolean register used serves as either input register, 
output register or auxiliary register.

Auxiliary Boolean registers are not needed to compute Boolean functions.
The question whether shorter instruction sequences are possible with the 
use of auxiliary Boolean registers was not answered in~\cite{BM13a}.
In~\cite{BM14e}, we showed that, in the case of the parity functions, 
shorter instruction sequences are possible with the use of an auxiliary 
Boolean register provided the instruction set is extended with
instructions to complement the content of auxiliary Boolean registers.
In the current paper, we consider all instructions for Boolean 
registers that are possible in the setting in which the work presented 
in~\cite{BM13a,BM14e} has been carried out and investigate instruction 
sequence size bounded functional completeness of instruction sets for 
Boolean registers.

Intuitively, a given instruction set for Boolean registers is 
$n$-size-bounded functionally complete if the effects of each possible 
instruction for Boolean registers can be obtained by an instruction 
sequence whose length is at most $n$ and that contains only instructions 
from the given instruction set for Boolean registers, forward jump 
instructions, and a termination instruction. 
A given instruction set for Boolean registers is functionally 
complete if it is $n$-size-bounded functionally complete for some $n$.

We have identified one of the $256$ smallest instruction sets for 
Boolean registers that is $1$-size-bounded functionally complete
(Corollary~\ref{corollary-bounded-complete-1}), and we have found that 
there is a large subset of this $1$-size-bounded functionally complete 
instruction set with the following properties:
(i)~each of its proper subsets that does not include the instructions to 
complement the content of Boolean registers, but includes the 
instructions to set and get the content of Boolean registers, is 
$4$-size-bounded functionally complete and not $3$-size-bounded 
functionally complete and
(ii)~each of its proper subsets that includes the instructions to 
complement the content of Boolean registers is 
$3$-size-bounded functionally complete and not $2$-size-bounded 
functionally complete (Corollary~\ref{corollary-bounded-complete-2}).

The use of a $1$-size-bounded functionally complete instruction 
set, such as the one referred to in the previous paragraph, gives rise 
to the smallest instruction sequence sizes.
However, the use of many instruction sets that are not $1$-size-bounded 
functionally complete, e.g.\ the ones referred to under~(i) above, gives 
rise to instruction sequence sizes that are at most $4$ times larger.

The work presented in~\cite{BM14e} triggered the work presented in this 
paper because of the choice to use an extension of the instruction set 
used in~\cite{BM13a}.
Since the results from the latter paper, with the exception of one 
auxiliary result, are concerned with upper bounds of instruction 
sequence size complexity, these results go through if instruction 
sequences may also contain instructions to complement the content of 
auxiliary Boolean registers.
However, for the work presented in~\cite{BM14e}, the instruction set did 
matter in the sense that we were not able to prove the main result of 
the paper, which is concerned with a lower bound of instruction sequence
size complexity, using the instruction set used in~\cite{BM13a}.
We consider the work presented in the current paper to be useful to 
adequately assess that result as it is.
We expect that it will turn out to be also useful to adequately assess 
results of future work that is concerned with lower bounds of 
instruction sequence size complexity.

Like the work presented in~\cite{BM13a,BM14e}, the work presented in 
this paper is carried out in the setting of \PGA\ (ProGram Algebra).
\PGA\ is an algebraic theory of single-pass instruction sequences that
was taken as the basis of the approach to the semantics of programming 
languages introduced in~\cite{BL02a}.
As a continuation of the work presented in~\cite{BL02a},
(i)~the notion of an instruction sequence was subjected to systematic
and precise analysis and
(ii)~issues relating to diverse subjects in computer science and 
computer engineering were rigorously investigated in the setting of 
\PGA. 
The subjects concerned include programming language expressiveness, 
computability, computational complexity, algorithm efficiency, 
algorithmic equivalence of programs, program verification, program 
compactness, micro-architecture, and probabilistic programming.
For a comprehensive survey of a large part of this work, 
see~\cite{BM12b}. 
An overview of all the work done to date in the setting of \PGA\ and 
open questions originating from this work can be found on~\cite{SiteIS}.

This paper is organized as follows.
First, we present the preliminaries to the work presented in this paper
(Sections~\ref{sect-PGA-and-BTA} and~\ref{sect-TSI}) and introduce the
possible instructions for Boolean registers 
(Section~\ref{sect-instr-br}).
Next, we define an equivalence relation on these instructions that 
identifies instructions that have the same effects 
(Section~\ref{sect-eqv-instr-br}) and study instruction sequence 
size bounded functional completeness of instruction sets for Boolean 
registers (Section~\ref{sect-bounded-complete}).
Finally, we make some concluding remarks (Section~\ref{sect-concl}).

Some familiarity with the basic notions related to algebraic theories 
and their models is assumed in this paper.
The relevant notions are explained in handbook chapters and books on the 
foundations of algebraic specification, 
e.g.~\cite{EM85a,ST99a,ST12a,Wir90a}.

The following should be mentioned in advance.
The set $\Bool$ is a set with two elements whose intended 
interpretations are the truth values \emph{false} and \emph{true}.
As is common practice, we represent the elements of $\Bool$ by the bits 
$0$ and $1$.
In line with generally accepted conventions, we use terminology based on
identification of the elements of $\Bool$ with their representation 
where appropriate.
For example, where a better link up with commonly used terminology is 
expected, the elements of $\Bool$ are loosely called bits and the 
elements of $\Bool^n$ are loosely called bit strings of length $n$.

The preliminaries to the work presented in this paper 
(Sections~\ref{sect-PGA-and-BTA} and~\ref{sect-TSI}) are almost the 
same as the preliminaries to the work presented in~\cite{BM14a} and 
earlier papers.
For this reason, there is some text overlap with those papers.
Apart from the preliminaries, the material in this paper is new.

\section{Program Algebra and Basic Thread Algebra}
\label{sect-PGA-and-BTA}

In this section, we give a survey of \PGA\ (ProGram Algebra) and \BTA\ 
(Basic Thread Algebra) and make precise in the setting of \BTA\ which 
behaviours are produced by the instruction sequences considered in \PGA\
under execution.
The greater part of this section originates from~\cite{BM13a}.
A comprehensive introduction to \PGA\ and \BTA, including examples, can 
among other things be found in~\cite{BM12b}.

In \PGA, it is assumed that there is a fixed but arbitrary set $\BInstr$
of \emph{basic instructions}.
The intuition is that the execution of a basic instruction may modify a 
state and produces a reply at its completion.
The possible replies are $\False$ and $\True$.
The actual reply is generally state-dependent.
The set $\BInstr$ is the basis for the set of instructions that may 
occur in the instruction sequences considered in \PGA.
The elements of the latter set are called \emph{primitive instructions}.
There are five kinds of primitive instructions:
\begin{itemize}
\item
for each $a \in \BInstr$, a \emph{plain basic instruction} $a$;
\item
for each $a \in \BInstr$, a \emph{positive test instruction} $\ptst{a}$;
\item
for each $a \in \BInstr$, a \emph{negative test instruction} $\ntst{a}$;
\item
for each $l \in \Nat$, a \emph{forward jump instruction} $\fjmp{l}$;
\item
a \emph{termination instruction} $\halt$.
\end{itemize}
We write $\PInstr$ for the set of all primitive instructions.

On execution of an instruction sequence, these primitive instructions
have the following effects:
\begin{itemize}
\item
the effect of a positive test instruction $\ptst{a}$ is that basic
instruction $a$ is executed and execution proceeds with the next
primitive instruction if $\True$ is produced and otherwise the next
primitive instruction is skipped and execution proceeds with the
primitive instruction following the skipped one --- if there is no
primitive instruction to proceed with,
inaction occurs;
\item
the effect of a negative test instruction $\ntst{a}$ is the same as
the effect of $\ptst{a}$, but with the role of the value produced
reversed;
\item
the effect of a plain basic instruction $a$ is the same as the effect
of $\ptst{a}$, but execution always proceeds as if $\True$ is produced;
\item
the effect of a forward jump instruction $\fjmp{l}$ is that execution
proceeds with the $l$th next primitive instruction --- if $l$ equals $0$ 
or there is no primitive instruction to proceed with, inaction occurs;
\item
the effect of the termination instruction $\halt$ is that execution 
terminates.
\end{itemize}

\PGA\ has one sort: the sort $\InSeq$ of \emph{instruction sequences}. 
We make this sort explicit to anticipate the need for many-sortedness
later on.
To build terms of sort $\InSeq$, \PGA\ has the following constants and 
operators:
\begin{itemize}
\item
for each $u \in \PInstr$, 
the \emph{instruction} constant $\const{u}{\InSeq}$\,;
\item
the binary \emph{concatenation} operator 
$\funct{\ph \conc \ph}{\InSeq \x \InSeq}{\InSeq}$\,;
\item
the unary \emph{repetition} operator 
$\funct{\ph\rep}{\InSeq}{\InSeq}$\,.
\end{itemize}
Terms of sort $\InSeq$ are built as usual in the one-sorted case.
We assume that there are infinitely many variables of sort $\InSeq$, 
including $X,Y,Z$.
We use infix notation for concatenation and postfix notation for
repetition.

A closed \PGA\ term is considered to denote a non-empty, finite or
eventually periodic infinite sequence of primitive instructions.%
\footnote
{An eventually periodic infinite sequence is an infinite sequence with
 only finitely many distinct suffixes.}
The instruction sequence denoted by a closed term of the form
$t \conc t'$ is the instruction sequence denoted by $t$
concatenated with the instruction sequence denoted by $t'$.
The instruction sequence denoted by a closed term of the form $t\rep$
is the instruction sequence denoted by $t$ concatenated infinitely
many times with itself.

Closed \PGA\ terms are considered equal if they represent the same
instruction sequence.
The axioms for instruction sequence equivalence are given in
Table~\ref{axioms-PGA}.%
\begin{table}[!t]
\caption{Axioms of \PGA}
\label{axioms-PGA}
\begin{eqntbl}
\begin{axcol}
(X \conc Y) \conc Z = X \conc (Y \conc Z)              & \axiom{PGA1} \\
(X^n)\rep = X\rep                                      & \axiom{PGA2} \\
X\rep \conc Y = X\rep                                  & \axiom{PGA3} \\
(X \conc Y)\rep = X \conc (Y \conc X)\rep              & \axiom{PGA4}
\end{axcol}
\end{eqntbl}
\end{table}
In this table, $n$ stands for an arbitrary natural number from 
$\Natpos$.%
\footnote
{We write $\Natpos$ for the set $\set{n \in \Nat \where n \geq 1}$ of
positive natural numbers.}
For each $n \in \Natpos$, the term $t^n$, where $t$ is a \PGA\ term, is 
defined by induction on $n$ as follows: $t^1 = t$, and 
$t^{n+1} = t \conc t^n$.

A typical model of \PGA\ is the model in which:
\begin{itemize}
\item
the domain is the set of all finite and eventually periodic infinite
sequences over the set $\PInstr$ of primitive instructions;
\item
the operation associated with ${} \conc {}$ is concatenation;
\item
the operation associated with ${}\rep$ is the operation ${}\srep$
defined as follows:

\begin{itemize}
\item
if $U$ is finite, then $U\srep$ is the unique infinite sequence $U'$ 
such that $U$ concatenated $n$ times with itself is a proper prefix of 
$U'$ for each $n \in \Nat$;
\item
if $U$ is infinite, then $U\srep$ is $U$.
\end{itemize}
\end{itemize}
It is immediately clear that this model has no proper subalgebra.
Moreover, we know from~\cite[Section~3.2.2]{BL02a} that the axioms of 
\PGA\ are complete with respect to satisfaction of equations between 
closed terms in this model. 
Hence, this model is an initial model of \PGA\ (see e.g.~\cite{ST99a}).

We confine ourselves to this model of \PGA\ for the interpretation of 
\PGA\ terms.
In the sequel, we use the term \emph{PGA instruction sequence} for the 
elements of the domain of this model.
Below, we will use \BTA\ to make precise which behaviours are produced 
by \PGA\ instruction sequences under execution.

In \BTA, it is assumed that a fixed but arbitrary set $\BAct$ of
\emph{basic actions} has been given.
The objects considered in \BTA\ are called threads.
A thread represents a behaviour which consists of performing basic 
actions in a sequential fashion.
Upon each basic action performed, a reply from an execution environment
determines how the thread proceeds.
The possible replies are the values $\False$ and $\True$.

\BTA\ has one sort: the sort $\Thr$ of \emph{threads}. 
We make this sort explicit to anticipate the need for many-sortedness
later on.
To build terms
of sort $\Thr$, \BTA\ has the following constants and operators:
\begin{itemize}
\item
the \emph{inaction} constant $\const{\DeadEnd}{\Thr}$;
\item
the \emph{termination} constant $\const{\Stop}{\Thr}$;
\item
for every $a \in \BAct$, the binary \emph{postconditional composition} 
operator $\funct{\pcc{\ph}{a}{\ph}}{\Thr \x \Thr}{\Thr}$.
\end{itemize}
Terms of sort $\Thr$ are built as usual in the one-sorted case. 
We assume that there are infinitely many variables of sort $\Thr$, 
including $x,y$.
We use infix notation for postconditional composition. 
We introduce \emph{basic action prefixing} as an abbreviation: 
$a \bapf t$, where $t$ is a \BTA\ term, abbreviates 
$\pcc{t}{a}{t}$.
We identify expressions of the form $a \bapf t$ with the \BTA\
term they stand for.

The thread denoted by a closed term of the form $\pcc{t}{a}{t'}$
will first perform $a$, and then proceed as the thread denoted by
$t$ if the reply from the execution environment is $\True$ and proceed
as the thread denoted by $t'$ if the reply from the execution
environment is $\False$. 
The thread denoted by $\Stop$ will do no more than terminate and the 
thread denoted by $\DeadEnd$ will become inactive.

Closed \BTA\ terms are considered equal if they are syntactically the
same.
Therefore, \BTA\ has no axioms.

Each closed \BTA\ term denotes a finite thread, i.e.\ a thread with a
finite upper bound to the number of basic actions that it can perform.
Infinite threads, i.e.\ threads without a finite upper bound to the
number of basic actions that it can perform, can be defined by means of 
a set of recursion equations (see e.g.~\cite{BM09k}).
We are only interested in models of \BTA\ in which sets of recursion 
equations have unique solutions, such as the projective limit model 
of \BTA\ presented in~\cite{BM12b}.

We confine ourselves to this model of \BTA, which has an initial model 
of \BTA\ as a submodel, for the interpretation of \BTA\ terms. 
In the sequel, we use the term \emph{BTA thread} or simply \emph{thread} 
for the elements of the domain of this model.

Regular threads, i.e.\ finite or infinite threads that can only be in a 
finite number of states, can be defined by means of a finite set of 
recursion equations.
Provided that the set $\BInstr$ of basic instructions is identified with 
the set $\BAct$ of basic actions, the behaviours produced by \PGA\ 
instruction sequences under execution are exactly the behaviours 
represented by regular threads and the behaviours produced by finite 
\PGA\ instruction sequences are exactly the behaviours represented by 
finite threads.

Henceforth, we will identify $\BInstr$ with $\BAct$.
Intuitively, this means that we will not distinguish the basic action 
that takes place when a basic instruction is executed from that 
basic instruction. 

We combine \PGA\ with \BTA, identifying $\BInstr$ with $\BAct$, and 
extend the combination with the \emph{thread extraction} operator 
$\funct{\extr{\ph}}{\InSeq}{\Thr}$, the axioms given in 
Table~\ref{axioms-thread-extr},%
\begin{table}[!tb]
\caption{Axioms for the thread extraction operator}
\label{axioms-thread-extr}
\begin{eqntbl}
\renewcommand{\arraystretch}{1.3}
\begin{eqncol}
\extr{a} = a \bapf \DeadEnd \\
\extr{a \conc X} = a \bapf \extr{X} \\
\extr{\ptst{a}} = a \bapf \DeadEnd \\
\extr{\ptst{a} \conc X} =
\pcc{\extr{X}}{a}{\extr{\fjmp{2} \conc X}} \\
\extr{\ntst{a}} = a \bapf \DeadEnd \\
\extr{\ntst{a} \conc X} =
\pcc{\extr{\fjmp{2} \conc X}}{a}{\extr{X}}
\end{eqncol}
\qquad
\begin{eqncol}
\extr{\fjmp{l}} = \DeadEnd \\
\extr{\fjmp{0} \conc X} = \DeadEnd \\
\extr{\fjmp{1} \conc X} = \extr{X} \\
\extr{\fjmp{l+2} \conc u} = \DeadEnd \\
\extr{\fjmp{l+2} \conc u \conc X} = \extr{\fjmp{l+1} \conc X} \\
\extr{\halt} = \Stop \\
\extr{\halt \conc X} = \Stop
\end{eqncol}
\end{eqntbl}
\end{table}
and the rule that $\extr{X} = \DeadEnd$ if $X$ has an infinite chain of 
forward jumps beginning at its first primitive instruction.%
\footnote
{This rule, which can be formalized using an auxiliary structural 
congruence predicate (see e.g.~\cite{BM12b}), is unnecessary when 
considering only finite \PGA\ instruction sequences.
}
In Table~\ref{axioms-thread-extr}, $a$ stands for an arbitrary basic 
instruction from $\BInstr$, $u$ stands for an arbitrary primitive 
instruction from $\PInstr$, and $l$ stands for an arbitrary natural 
number from $\Nat$.
For each closed \PGA\ term $t$, $\extr{t}$ denotes the behaviour  
produced by the instruction sequence denoted by $t$ under execution.

\section{Interaction of Threads with Services}
\label{sect-TSI}

Services are objects that represent the behaviours exhibited by 
components of execution environments of instruction sequences at a high 
level of abstraction.
A service is able to process certain methods.
For the purpose of the extension of \BTA\ that will be presented in this 
section, it is sufficient to know the following about methods:
(i)~the processing of a method by a service may involve a change of the 
service and
(ii)~at completion of the processing of a method by a service, the 
service produces a reply value. 
The possible reply values are $\False$ and $\True$.
Execution environments are considered to provide a family of 
uniquely-named services.

A thread may interact with the named services from the service family 
provided by an execution environment.
That is, a thread may perform a basic action for the purpose of 
requesting a named service to process a method and to return a reply 
value at completion of the processing of the method.
In this section, we give a survey of the extension of \BTA\ with 
services, service families, a composition operator for service families, 
and operators that are concerned with this kind of interaction.
This section originates from~\cite{BM09k}.
A comprehensive introduction to the presented extension of \BTA, 
including examples, can among other things be found in~\cite{BM12b}.

First, we introduce an algebraic theory of service families called 
\SFA\ (Service Family Algebra).
In \SFA, it is assumed that a fixed but arbitrary set $\Meth$ of 
\emph{methods} has been given.
Moreover, the following is assumed with respect to services:
\begin{itemize}
\item
a signature $\Sig{\Services}$ has been given that includes the following
sorts:
\begin{itemize}
\item
the sort $\Serv$ of
\emph{services};
\item
the sort $\Repl$ of \emph{replies};
\end{itemize}
and the following constants and operators:
\begin{itemize}
\item
the
\emph{empty service} constant $\const{\emptyserv}{\Serv}$;
\item
the \emph{reply} constants $\const{\False,\True,\Div}{\Repl}$;
\item
for each $m \in \Meth$, the
\emph{derived service} operator $\funct{\derive{m}}{\Serv}{\Serv}$;
\item
for each $m \in \Meth$, the
\emph{service reply} operator $\funct{\sreply{m}}{\Serv}{\Repl}$;
\end{itemize}
\item
a $\Sig{\Services}$-algebra $\ServAlg$ that has no proper subalgebra has 
been given in which the following holds:
\begin{itemize}
\item
$\False \neq \True$, $\True \neq \Div$, $\Div \neq \False$;
\item
for each $m \in \Meth$,
$\derive{m}(z) = \emptyserv \Liff \sreply{m}(z) = \Div$.
\end{itemize}
\end{itemize}

The intuition concerning $\derive{m}$ and $\sreply{m}$ is that on a
request to service $s$ to process method $m$:
\begin{itemize}
\item
if $\sreply{m}(s) \neq \Div$, $s$ processes $m$, produces the reply
$\sreply{m}(s)$, and then proceeds as $\derive{m}(s)$;
\item
if $\sreply{m}(s) = \Div$, $s$ is not able to process method $m$ and
proceeds as $\emptyserv$.
\end{itemize}
The empty service $\emptyserv$ itself is unable to process any method.

It is also assumed that a fixed but arbitrary set $\Foci$ of
\emph{foci} has been given.
Foci play the role of names of services in a service family. 

\SFA\ has the sorts, constants and operators from $\Sig{\Services}$ and
in addition the sort $\ServFam$ of \emph{service families} and the 
following constant and operators:
\begin{itemize}
\item
the
\emph{empty service family} constant $\const{\emptysf}{\ServFam}$;
\item
for each $f \in \Foci$, the unary
\emph{singleton service family} operator
$\funct{\mathop{f{.}} \ph}{\Serv}{\ServFam}$;
\item
the binary
\emph{service family composition} operator
$\funct{\ph \sfcomp \ph}{\ServFam \x \ServFam}{\ServFam}$;
\item
for each $F \subseteq \Foci$, the unary
\emph{encapsulation} operator $\funct{\encap{F}}{\ServFam}{\ServFam}$.
\end{itemize}
We assume that there are infinitely many variables of sort $\Serv$,
including $z$, and infinitely many variables of sort $\ServFam$,
including $u,v,w$.
Terms are built as usual in the many-sorted case 
(see e.g.~\cite{ST12a}).
We use prefix notation for the singleton service family operators and
infix notation for the service family composition operator.
We write $\Sfcomp{i = 1}{n} t_i$, where $t_1,\ldots,t_n$ are
terms of sort $\ServFam$, for the term
$t_1 \sfcomp \ldots \sfcomp t_n$.

The service family denoted by $\emptysf$ is the empty service family.
The service family denoted by a closed term of the form $f.t$ consists 
of one named service only, the service concerned is the service denoted 
by $t$, and it is named $f$.
The service family denoted by a closed term of the form
$t \sfcomp t'$ consists of all named services that belong to either the
service family denoted by $t$ or the service family denoted by $t'$.
In the case where a named service from the service family denoted by
$t$ and a named service from the service family denoted by $t'$ have
the same name, they collapse to an empty service with the name
concerned.
The service family denoted by a closed term of the form $\encap{F}(t)$ 
consists of all named services with a name not in $F$ that belong to
the service family denoted by $t$.

The axioms of \SFA\ are given in 
Table~\ref{axioms-SFA}.%
\begin{table}[!t]
\caption{Axioms of \SFA}
\label{axioms-SFA}
{
\begin{eqntbl}
\begin{axcol}
u \sfcomp \emptysf = u                                 & \axiom{SFC1} \\
u \sfcomp v = v \sfcomp u                              & \axiom{SFC2} \\
(u \sfcomp v) \sfcomp w = u \sfcomp (v \sfcomp w)      & \axiom{SFC3} \\
f.z \sfcomp f.z' = f.\emptyserv                        & \axiom{SFC4}
\end{axcol}
\qquad
\begin{saxcol}
\encap{F}(\emptysf) = \emptysf                       & & \axiom{SFE1} \\
\encap{F}(f.z) = \emptysf            & \mif f \in F    & \axiom{SFE2} \\
\encap{F}(f.z) = f.z                 & \mif f \notin F & \axiom{SFE3} \\
\multicolumn{2}{@{}l@{\quad}}
 {\encap{F}(u \sfcomp v) =
  \encap{F}(u) \sfcomp \encap{F}(v)}                   & \axiom{SFE4}
\end{saxcol}
\end{eqntbl}
}
\end{table}
In this table, $f$ stands for an arbitrary focus from $\Foci$ and
$F$ stands for an arbitrary subset of $\Foci$.
These axioms simply formalize the informal explanation given
above.

For the set $\BAct$ of basic actions, we now take 
$\set{f.m \where f \in \Foci, m \in \Meth}$.
Performing a basic action $f.m$ is taken as making a request to the
service named $f$ to process method $m$.

We combine \BTA\ with \SFA\ and extend the combination with the 
following operators:
\pagebreak[2]
\begin{itemize}
\item
the binary \emph{abstracting use} operator
$\funct{\ph \sfause \ph}{\Thr \x \ServFam}{\Thr}$;
\item
the binary \emph{apply} operator
$\funct{\ph \sfapply \ph}{\Thr \x \ServFam}{\ServFam}$;
\end{itemize}
and the axioms given in Tables~\ref{axioms-abstracting-use} 
and~\ref{axioms-apply}.%
\begin{table}[!t]
\caption{Axioms for the abstracting use operator}
\label{axioms-abstracting-use}
\begin{eqntbl}
\begin{saxcol}
\Stop  \sfause u = \Stop                              & & \axiom{AU1} \\
\DeadEnd \sfause u = \DeadEnd                         & & \axiom{AU2} \\
(\pcc{x}{f.m}{y}) \sfause \encap{\set{f}}(u) =
\pcc{(x \sfause \encap{\set{f}}(u))}
 {f.m}{(y \sfause \encap{\set{f}}(u))}                & & \axiom{AU3} \\
(\pcc{x}{f.m}{y}) \sfause (f.t \sfcomp \encap{\set{f}}(u)) =
x \sfause (f.\derive{m}t \sfcomp \encap{\set{f}}(u))
                          & \mif \sreply{m}(t) = \True  & \axiom{AU4} \\
(\pcc{x}{f.m}{y}) \sfause (f.t \sfcomp \encap{\set{f}}(u)) =
y \sfause (f.\derive{m}t \sfcomp \encap{\set{f}}(u))
                          & \mif \sreply{m}(t) = \False & \axiom{AU5} \\
(\pcc{x}{f.m}{y}) \sfause (f.t \sfcomp \encap{\set{f}}(u)) = \DeadEnd
                          & \mif \sreply{m}(t) = \Div   & \axiom{AU6}
\end{saxcol}
\end{eqntbl}
\end{table}
\begin{table}[!t]
\caption{Axioms for the apply operator}
\label{axioms-apply}
\begin{eqntbl}
\begin{saxcol}
\Stop  \sfapply u = u                                  & & \axiom{A1} \\
\DeadEnd \sfapply u = \emptysf                         & & \axiom{A2} \\
(\pcc{x}{f.m}{y}) \sfapply \encap{\set{f}}(u) = \emptysf
                                                       & & \axiom{A3} \\
(\pcc{x}{f.m}{y}) \sfapply (f.t \sfcomp \encap{\set{f}}(u)) =
x \sfapply (f.\derive{m}t \sfcomp \encap{\set{f}}(u))
                           & \mif \sreply{m}(t) = \True  & \axiom{A4} \\
(\pcc{x}{f.m}{y}) \sfapply (f.t \sfcomp \encap{\set{f}}(u)) =
y \sfapply (f.\derive{m}t \sfcomp \encap{\set{f}}(u))
                           & \mif \sreply{m}(t) = \False & \axiom{A5} \\
(\pcc{x}{f.m}{y}) \sfapply (f.t \sfcomp \encap{\set{f}}(u)) = \emptysf
                           & \mif \sreply{m}(t) = \Div   & \axiom{A6}
\end{saxcol}
\end{eqntbl}
\end{table}
In these tables, $f$ stands for an arbitrary focus from $\Foci$, $m$ 
stands for an arbitrary method from $\Meth$, and $t$ stands for an 
arbitrary term of sort $\Serv$.
The axioms formalize the informal explanation given below and in 
addition stipulate what is the result of abstracting use and apply if 
inappropriate foci or methods are involved.
We use infix notation for the abstracting use and apply operators.

The thread denoted by a closed term of the form $t \sfause t'$ and the
service family denoted by a closed term of the form $t \sfapply t'$ are
the thread and service family, respectively, that result from processing
the method of each basic action performed by the thread denoted by $t$
by the service in the service family denoted by $t'$ with the focus
of the basic action as its name if such a service exists.
When the method of a basic action performed by a thread is processed by
a service, the service changes in accordance with the method concerned
and the thread reduces to one of the two threads that it can possibly 
proceed with dependent on the reply value produced by the service.

The projective limit model of the extension of the combination of \BTA\ 
and \SFA\ with the abstracting use operator, the apply operator, and the 
axioms for these operators is a reduct of the projective limit model 
presented in~\cite[Section~3.1.9]{BM12b}. 
The reduct of this model to the constants and operators of \BTA\ is
the projective limit model of \BTA.

\section{Instructions for Boolean Registers}
\label{sect-instr-br}

The primitive instructions that concern us in the remainder of this 
paper are primitive instructions for Boolean registers.
We introduce in this section the possible primitive instructions for 
Boolean registers.

It is assumed that, for each $\funct{p,q}{\Bool}{\Bool}$, 
$\mbr{p}{q} \in \Meth$.
These methods can be explained as follows:
\begin{quote}
when $\mbr{p}{q}$ is processed by a Boolean register service whose 
register content is $b$, the reply is $p(b)$ and the register content 
becomes $q(b)$.
\end{quote}
We write $\Methbr$ for the set 
$\set{\mbr{p}{q} \where \funct{p,q}{\Bool}{\Bool}}$.
Every method that a Boolean register service could possibly process is a 
method from $\Methbr$.

For $\Sig{\Services}$, we take the signature that consists of the sorts,
constants and operators that are mentioned in the assumptions with
respect to services made in Section~\ref{sect-TSI} and a constant 
$\BR^M_b$ for each $M \subseteq \Methbr$ and $b \in \Bool$.
Informally, $\BR^M_b$ denotes the Boolean register service with register
content $b$ that is able to process precisely all methods from $M$.

For $\ServAlg$, we take the $\Sig{\Services}$-algebra that has no 
proper subalgebra and that satisfies the conditions that are mentioned 
in the assumptions with respect to services made in 
Section~\ref{sect-TSI} and the following conditions for each 
$M \subseteq \Methbr$ and $b \in \Bool$:
\begin{ldispl}
\begin{gceqns}
\derive{\mbr{p}{q}}(\BR^M_b) = \BR^M_{q(b)} & \mif \mbr{p}{q} \in M\;,
\eqnsep
\sreply{\mbr{p}{q}}(\BR^M_b) = p(b)         & \mif \mbr{p}{q} \in M\;,
\end{gceqns}
\qquad
\begin{gceqns}
\derive{m}(\BR^M_b) = \emptyserv & \mif m \notin M\;,
\eqnsep
\sreply{m}(\BR^M_b) = \Div       & \mif m \notin M\;.
\end{gceqns}
\end{ldispl}%

$\Bool \to \Bool$, the set of all unary Boolean functions, consists of 
the following four functions:
\begin{itemize}
\item
the function $\FFunc$, satisfying 
$\FFunc(\False) = \False$ and $\FFunc(\True) = \False$;
\item
the function $\TFunc$, satisfying 
$\TFunc(\False) = \True$ and $\TFunc(\True) = \True$;
\item
the function $\IFunc$, satisfying 
$\IFunc(\False) = \False$ and $\IFunc(\True) = \True$;
\item
the function $\CFunc$, satisfying 
$\CFunc(\False) = \True$ and $\CFunc(\True) = \False$.
\end{itemize}
In~\cite{BM13a}, we actually used the methods $\mbr{\FFunc}{\FFunc}$, 
$\mbr{\TFunc}{\TFunc}$, and $\mbr{\IFunc}{\IFunc}$, but denoted them by 
$\setbr{0}$, $\setbr{1}$ and $\getbr$, respectively.
In~\cite{BM14e}, we actually used, in addition to these methods, the 
method $\mbr{\CFunc}{\CFunc}$, but denoted it by $\negbr$.

We define, for each $M \subseteq \Methbr$, the following sets:
\begin{ldispl}
\BIbr(M) = \set{f.m \where f \in \Foci \Land m \in M}\;,
\eqnsep
\PIbr(M) = 
\BIbr(M) \union \set{\ptst{a} \where a \in \BIbr(M)} \union
\set{\ntst{a} \where a \in \BIbr(M)}\;.
\end{ldispl}%
$\BIbr(\Methbr)$ consists of 16 basic actions per focus and 
$\PIbr(\Methbr)$ consists of 48 primitive instructions per focus.

For Boolean registers that serve as input register, we used 
in~\cite{BM13a,BM14e} only primitive instructions from 
$\PIbr(\set{\mbr{\IFunc}{\IFunc}})$. 
For Boolean registers that serve as output register, we used 
in~\cite{BM13a,BM14e} only primitive instructions  from 
$\PIbr(\set{\mbr{\FFunc}{\FFunc},\mbr{\TFunc}{\TFunc}})$.
For Boolean registers that serve as auxiliary register, we used 
in~\cite{BM13a} only primitive instructions from
$\PIbr(\set{\mbr{\FFunc}{\FFunc},\mbr{\TFunc}{\TFunc},
                                 \mbr{\IFunc}{\IFunc}})$
and in~\cite{BM14e} only primitive instructions from 
$\PIbr(\set{\mbr{\FFunc}{\FFunc},\mbr{\TFunc}{\TFunc},
            \mbr{\IFunc}{\IFunc},\mbr{\CFunc}{\CFunc}})$.
However, in the case of auxiliary registers,  
other possible instruction sets are eligible.
In Section~\ref{sect-bounded-complete}, we study instruction sequence 
size-bounded functional completeness of instruction sets for Boolean 
registers.
We expect that the results of that study will turn out to be useful to 
adequately assess results of work that is concerned with lower bounds of 
instruction sequence size complexity in cases where auxiliary Boolean 
registers may be used.

We write $\ISbr(M)$, where $M \subseteq \Methbr$, for the set of all 
finite \PGA\ instruction sequences in the case where $\BIbr(M)$ is taken
for the set $\BInstr$ of basic instructions.

\section{Equivalence of Instructions for Boolean Registers}
\label{sect-eqv-instr-br}

There exists a model of the extension of the combination of \PGA, \BTA, 
and \SFA\ with the thread extraction operator, the abstracting use 
operator, the apply operator, and the axioms for these operators such 
that the initial model of \PGA\ is its reduct to the signature of \PGA\
and the projective limit model of the extension of the combination of 
\BTA\ and \SFA\ with the abstracting use operator, the apply operator, 
and the axioms for these operators is its reduct to the signature of 
this theory.
This follows from the disjointness of the signatures concerned by the 
amalgamation result about expansions presented as Theorem 6.1.1 
in~\cite{Hod93a} (adapted to the many-sorted case).

Henceforth, we work in the model just mentioned, and denote the 
interpretations of constants and operators in it by the constants and 
operators themselves.
However, we could work in any model for which the axioms are complete 
with respect to satisfaction of equations between closed terms.

On execution of an instruction sequence, different primitive 
instructions from $\PIbr(\Methbr)$ do not always have different effects.
We define an equivalence on $\PIbr(\Methbr)$ that identifies primitive 
instructions if they have the same effects.

Let $u,v \in \PIbr(\Methbr)$. 
Then $u$ and $v$ are \emph{effectually equivalent}, written $u \eeqv v$, 
if there exists an $f \in \Foci$ such that, for each $b \in \Bool$ and 
$n \in \set{1,2}$:
\begin{center}
\lststretch
\begin{tabular}[t]{@{}l@{}}
$\extr{u \conc \halt^n} \sfause f.\BR^\Methbr_b =
 \extr{v \conc \halt^n} \sfause f.\BR^\Methbr_b\;$,
\\
$\extr{u \conc \halt^n} \sfapply f.\BR^\Methbr_b \hsp{.2} {} =
 \extr{v \conc \halt^n} \sfapply f.\BR^\Methbr_b\;$.
\end{tabular}
\end{center}

Let $u,v \in \PIbr(\Methbr)$ be such that $u \neq v$, and 
let $f \in \Foci$ be such that, for some $m \in \Methbr$, 
$v \equiv f.m$ or $v \equiv \ptst{f.m}$ or $v \equiv \ntst{f.m}$.
Then $u \eeqv v$ only if 
$\extr{u \conc \halt^n} \sfause f.\BR^\Methbr_b =
 \extr{v \conc \halt^n} \sfause f.\BR^\Methbr_b$ and
$\extr{u \conc \halt^n} \sfapply f.\BR^\Methbr_b =
 \extr{v \conc \halt^n} \sfapply f.\BR^\Methbr_b$ 
for each $b \in \Bool$ and $n \in \set{1,2}$. 
From this and the definition of $\eeqv$, it follows immediately that 
$\eeqv$ is transitive.
Moreover, it follows immediately from the definition of $\eeqv$ that 
$\eeqv$ is reflexive and symmetric.
Hence, $\eeqv$ is an equivalence relation indeed.

Replacement of primitive instructions in an instruction sequence by 
effectually equivalent ones does not change the functionality of the
instruction sequence on execution.

Let $X,Y \in \ISbr(\Methbr)$. 
Then $X$ and $Y$ are \emph{functionally equivalent}, 
written $X \feqv Y$, if, for some $n \in \Natpos$,
there exist $f_1,\ldots,f_n \in \Foci$ such that, 
for each $b_1,\ldots,b_n \in \Bool$:
\begin{center}
\lststretch
\begin{tabular}[t]{@{}l@{}}
$\extr{X} \sfause \Sfcomp{i = 1}{n} f_i.\BR^\Methbr_{b_i} = \Stop$ or
$\extr{X} \sfause \Sfcomp{i = 1}{n} f_i.\BR^\Methbr_{b_i} = \DeadEnd$,
\\
$\extr{X} \sfause \Sfcomp{i = 1}{n} f_i.\BR^\Methbr_{b_i} =
 \extr{Y} \sfause \Sfcomp{i = 1}{n} f_i.\BR^\Methbr_{b_i}$,
\\
$\extr{X} \sfapply \Sfcomp{i = 1}{n} f_i.\BR^\Methbr_{b_i} \hsp{.2} {} =
 \extr{Y} \sfapply \Sfcomp{i = 1}{n} f_i.\BR^\Methbr_{b_i}$.
\end{tabular}
\end{center}

The proof that $\feqv$ is an equivalence relation goes along similar 
lines as the proof that $\eeqv$ is an equivalence relation.
Here, $X \feqv Y$ only if the equations from the definition hold in the 
case where we take the foci of primitive instructions from 
$\PIbr(\Methbr)$ that occur in $Y$ for $f_1,\ldots,f_n$.

\begin{proposition}
\label{proposition-eff-equiv}
Let $u,v \in \PIbr(\Methbr)$, and
let $X,Y \in \ISbr(\Methbr)$ be such that $Y$ is $X$ with every 
occurrence of $u$ replaced by $v$. 
Then $u \eeqv v$ implies $X \feqv Y$.
\end{proposition}
\begin{proof}
It is easily proved by induction on the length of $X$ that $u \eeqv v$ 
implies, for each $l,n \in \Nat$, 
$\fjmp{l} \conc X \conc \halt^n \feqv \fjmp{l} \conc Y \conc \halt^n$.%
\footnote
{We use the convention that $t'\conc t^0$ stands for $t'$.}
From this, it follows immediately that $u \eeqv v$ implies $X \feqv Y$.
\qed
\end{proof}

Axioms for effectual equivalence are given in 
Table~\ref{axioms-eff-equiv}.
\begin{table}[!t]
\caption{Axioms for effectual equivalence}
\label{axioms-eff-equiv}
\begin{eqntbl}
\begin{eqncol}
\ptst{f.\mbr{\FFunc}{p}} \eeqv \ntst{f.\mbr{\TFunc}{p}} \\
\ptst{f.\mbr{\TFunc}{p}} \eeqv \ntst{f.\mbr{\FFunc}{p}} \\
\ptst{f.\mbr{\IFunc}{p}} \eeqv \ntst{f.\mbr{\CFunc}{p}} \\
\ptst{f.\mbr{\CFunc}{p}} \eeqv \ntst{f.\mbr{\IFunc}{p}} 
\end{eqncol}
\qquad
\begin{eqncol}
\ptst{f.\mbr{\TFunc}{p}} \eeqv f.\mbr{q}{p} 
\end{eqncol}
\qquad
\begin{eqncol}
u \eeqv u \\
u \eeqv v \Limpl v \eeqv u \\
u \eeqv v \Land v \eeqv w \Limpl u \eeqv w 
\end{eqncol}
\end{eqntbl}
\end{table}
In this table, $f$ stands for an arbitrary focus from $\Foci$, $p$ and 
$q$ stand for arbitrary functions from $\Bool \to \Bool$, and $u$, $v$, 
and $w$ stand for arbitrary primitive instructions from 
$\PIbr(\Methbr)$. 
Moreover, we use $\eeqv$ in this table as a predicate symbol (and not as 
the symbol that denotes the effectual equivalence relation on 
$\PIbr(\Methbr)$ defined above).
\begin{theorem}
\label{theorem-eff-equiv}
The axioms in Table~\ref{axioms-eff-equiv} are sound and complete for 
the effectual equivalence relation on $\PIbr(\Methbr)$ defined above.
\end{theorem}
\begin{proof}
The soundness of the axioms follows immediately from the definition of 
effectual equivalence, using the conditions on $\ServAlg$ laid down in 
Section~\ref{sect-instr-br}.

The following conclusions can be drawn from the definition of effectual 
equivalence:
\begin{flushleft}
\lststretch
\begin{tabular}[t]{@{\hsp{5}}l@{\hsp{.7}}l@{}}
\textup{(a)} &
$\ptst{f.\mbr{p}{q}} \eeqv f.\mbr{p'}{q'} \Limpl  
 p = \TFunc \Land q = q'$\,;
\\
\textup{(b)} &
$\ntst{f.\mbr{p}{q}} \eeqv f.\mbr{p'}{q'} \Limpl 
 p = \FFunc \Land q = q'$\,;
\\
\textup{(c)} &
$\ptst{f.\mbr{p}{q}} \eeqv \ntst{f.\mbr{p'}{q'}} \Limpl 
 p = C(p') \Land q = q'$\,,
\\ &
where $C(\FFunc) = \TFunc$, $C(\TFunc) = \FFunc$, $C(\IFunc) = \CFunc$, 
$C(\CFunc) = \IFunc$.
\end{tabular}
\end{flushleft}
The completeness of the axioms follows easily by case distinction 
between the different forms that a formula $u \eeqv v$ can take, making
use of~(a), (b), and~(c).
\qed
\end{proof}

The equivalence classes of $\PIbr(\Methbr)$ with respect to $\eeqv$ are
the following for each $f \in \Foci$:
\begin{ldispl}
\set{\ul{\ptst{f.\mbr{\FFunc}{\FFunc}}},
     \ntst{f.\mbr{\TFunc}{\FFunc}}}\;, \\
\set{\ptst{f.\mbr{\FFunc}{\TFunc}},
     \ul{\ntst{f.\mbr{\TFunc}{\TFunc}}}}\;, \\
\set{\ptst{f.\mbr{\FFunc}{\IFunc}},
     \ul{\ntst{f.\mbr{\TFunc}{\IFunc}}}}\;, \\
\set{\ptst{f.\mbr{\FFunc}{\CFunc}},
     \ul{\ntst{f.\mbr{\TFunc}{\CFunc}}}}\;, 
\\
\set{\ptst{f.\mbr{\TFunc}{\FFunc}},\ntst{f.\mbr{\FFunc}{\FFunc}},
     \ul{f.\mbr{\FFunc}{\FFunc}},f.\mbr{\TFunc}{\FFunc},
     f.\mbr{\IFunc}{\FFunc},f.\mbr{\CFunc}{\FFunc}}\;, \\
\set{\ptst{f.\mbr{\TFunc}{\TFunc}},\ntst{f.\mbr{\FFunc}{\TFunc}},
     f.\mbr{\FFunc}{\TFunc},\ul{f.\mbr{\TFunc}{\TFunc}},
     f.\mbr{\IFunc}{\TFunc},f.\mbr{\CFunc}{\TFunc}}\;, \\
\set{\ptst{f.\mbr{\TFunc}{\IFunc}},\ntst{f.\mbr{\FFunc}{\IFunc}},
     f.\mbr{\FFunc}{\IFunc},f.\mbr{\TFunc}{\IFunc},
     \ul{f.\mbr{\IFunc}{\IFunc}},f.\mbr{\CFunc}{\IFunc}}\;, \\
\set{\ptst{f.\mbr{\TFunc}{\CFunc}},\ntst{f.\mbr{\FFunc}{\CFunc}},
     f.\mbr{\FFunc}{\CFunc},f.\mbr{\TFunc}{\CFunc},
     f.\mbr{\IFunc}{\CFunc},\ul{f.\mbr{\CFunc}{\CFunc}}}\;, 
\\
\set{\ul{\ptst{f.\mbr{\IFunc}{\FFunc}}},
     \ntst{f.\mbr{\CFunc}{\FFunc}}}\;, \\
\set{\ul{\ptst{f.\mbr{\IFunc}{\TFunc}}},
     \ntst{f.\mbr{\CFunc}{\TFunc}}}\;, \\
\set{\ul{\ptst{f.\mbr{\IFunc}{\IFunc}}},
     \ntst{f.\mbr{\CFunc}{\IFunc}}}\;, \\
\set{\ptst{f.\mbr{\IFunc}{\CFunc}},
     \ul{\ntst{f.\mbr{\CFunc}{\CFunc}}}}\;, 
\\
\set{\ptst{f.\mbr{\CFunc}{\FFunc}},
     \ul{\ntst{f.\mbr{\IFunc}{\FFunc}}}}\;, \\
\set{\ptst{f.\mbr{\CFunc}{\TFunc}},
     \ul{\ntst{f.\mbr{\IFunc}{\TFunc}}}}\;, \\
\set{\ptst{f.\mbr{\CFunc}{\IFunc}},
     \ul{\ntst{f.\mbr{\IFunc}{\IFunc}}}}\;, \\
\set{\ul{\ptst{f.\mbr{\CFunc}{\CFunc}}},
     \ntst{f.\mbr{\IFunc}{\CFunc}}}\;.
\end{ldispl}%
We have underlined one representative of each equivalence class in order
to refer to them easily in the proof of the following theorem.

\begin{theorem}
\label{theorem-min-meth-set}
\mbox{}
\begin{itemize}
\item[\textup{(1)}]
The set 
$\set{\mbr{\FFunc}{\FFunc},\mbr{\TFunc}{\TFunc},
      \mbr{\IFunc}{\IFunc},\mbr{\CFunc}{\CFunc},
      \mbr{\IFunc}{\FFunc},\mbr{\IFunc}{\TFunc},
      \mbr{\TFunc}{\IFunc},\mbr{\TFunc}{\CFunc}}$ 
is a minimal $M \subseteq \Methbr$ such that $\PIbr(M)$ contains at 
least one representative from each of the equivalence classes of 
$\PIbr(\Methbr)$ with respect to $\eeqv$.
\item[\textup{(2)}]
Each minimal $M \subseteq \Methbr$ such that $\PIbr(M)$ contains at 
least one representative from each of the equivalence classes of 
$\PIbr(\Methbr)$ with respect to $\eeqv$ consists of $8$ methods.
\end{itemize}
\end{theorem}
\begin{proof}
By uniformity, it is sufficient to look at the equivalence classes of 
$\PIbr(\Methbr)$ for an arbitrary focus from $\Foci$.
\begin{itemize}
\item[(1)]
Let 
$M'= \set{\mbr{\FFunc}{\FFunc},\mbr{\TFunc}{\TFunc},
          \mbr{\IFunc}{\IFunc},\mbr{\CFunc}{\CFunc},
          \mbr{\IFunc}{\FFunc},\mbr{\IFunc}{\TFunc},
          \mbr{\TFunc}{\IFunc},\mbr{\TFunc}{\CFunc}}$.
Then the representatives of the different equivalence classes of 
$\PIbr(\Methbr)$ that are underlined above belong to the set
$\PIbr(M')$.
Moreover, each method from $M'$ occurs in a primitive instruction from 
$\PIbr(\Methbr)$ that belongs to an equivalence class of 
$\PIbr(\Methbr)$ that contains only one other primitive instruction, but 
the method that occurs in this other primitive instruction is not from
$M'$.
Hence $M'$ is minimal.
\item[(2)]
First, we consider the first and last four equivalence classes above.
Each of them consists of two primitive instructions.
Each method that occurs in the primitive instructions from one of them 
does not occur in the primitive instructions from another of them.
Consequently, exactly eight methods are needed for representatives from 
these equivalence classes.
Next, we consider the remaining eight equivalence classes.
For each of them, the methods that occur in the primitive instructions 
from it include the methods that occur in the primitive instructions 
from one of the equivalence classes that we considered first.
Consequently, no additional methods are needed for representatives from 
the remaining equivalence classes.
Hence, exactly eight methods are needed for representatives from all 
equivalence classes. \qed
\end{itemize}
\end{proof}
Theorem~\ref{theorem-min-meth-set} tells us that each primitive 
instruction from $\PIbr(\Methbr)$ has the same effects as one with a 
method from 
$\set{\mbr{\FFunc}{\FFunc},\mbr{\TFunc}{\TFunc},
      \mbr{\IFunc}{\IFunc},\mbr{\CFunc}{\CFunc},
      \mbr{\IFunc}{\FFunc},\mbr{\IFunc}{\TFunc},
      \mbr{\TFunc}{\IFunc},\mbr{\TFunc}{\CFunc}}$
and that there does not exist a smaller set with this 
property.
The methods that we used in~\cite{BM13a,BM14e} are included in this set.

We have the following corollary of the proof of part~(2) of
Theorem~\ref{theorem-min-meth-set}.
\begin{corollary}
\label{corollary-min-meth-set}
There exist $256$ minimal $M \subseteq \Methbr$ such that $\PIbr(M)$ 
contains at least one representative from each of the equivalence 
classes of $\PIbr(\Methbr)$ with respect to $\eeqv$.
\end{corollary}

\section{Bounded Functional Completeness of Instruction Sets}
\label{sect-bounded-complete}

Not all methods from the minimal set mentioned in 
Theorem~\ref{theorem-min-meth-set} are needed to obtain the effects of 
each primitive instruction from $\PIbr(\Methbr)$ in the case where 
instruction sequences instead of instructions are used to obtain the 
effects. 
In this section, we look at the case where instruction sequences are 
used.  
We begin by defining the notion of $k$-size-bounded functional 
completeness.

Let $M \subseteq \Methbr$ and $k \in \Natpos$.  
Then the instruction set $\PIbr(M)$ is 
\emph{$k$-size-bounded functionally complete} if 
there exists a function $\funct{\psi}{\PIbr(\Methbr)}{\ISbr(M)}$ such 
that, for each $u \in \PIbr(\Methbr)$, $\len(\psi(u)) \leq k$ and
there exists an $f \in \Foci$ such that, 
for each $b \in \Bool$ and $n \in \Nat$:
\begin{center}
\lststretch
\begin{tabular}[t]{@{}l@{}}
$\extr{u \conc \halt^n} \sfause f.\BR^\Methbr_b =
 \extr{\psi(u) \conc \halt^n} \sfause f.\BR^\Methbr_b$,
\\
$\extr{u \conc \halt^n} \sfapply f.\BR^\Methbr_b \hsp{.2} {} =
 \extr{\psi(u) \conc \halt^n} \sfapply f.\BR^\Methbr_b$. 
\end{tabular}
\end{center}
$\PIbr(M)$ is called 
\emph{strictly $k$-size-bounded functionally complete} if $\PIbr(M)$ is 
$k$-size-bounded functionally complete and there does not exist a 
$k' < k$ such that $\PIbr(M)$ is $k'$-size-bounded functionally 
complete.

The following proposition illustrates the relevance of the notion of 
$k$-size-bounded functionally completeness.
\begin{proposition}
\label{proposition-bounded-complete}
Let $M \subseteq \Methbr$ and $k \in \Natpos$.
Let $\funct{\psi}{\PIbr(\Methbr)}{\ISbr(M)}$ be as in the definition of
$k$-size-bounded functional completeness given above. 
Let $\funct{\psi'}{\ISbr(\Methbr)}{\ISbr(M)}$ be such that
$\psi'(u_1 \conc \ldots \conc u_n) = u_1' \conc \ldots \conc u_n'$, 
where
\begin{center}
\lststretch
\begin{tabular}[t]{@{}l@{\,\,}l@{}}
$u_i' \equiv \halt$ & if $u_i \equiv \halt$;
\\
$u_i' \equiv \fjmp{l'}$ \\ \,\,\,\, with 
$l' = 
l + \sum_{j \in \set{i,\ldots,i+l-1}
          \mathrm{\,s.t.\,} u_j \in \PIbr(\Methbr)}
     (\len(\psi(u_j)) - 1)$ & if $u_i \equiv \fjmp{l}$;
\\
$u_i' \equiv \psi(u_i)$ & otherwise.
\end{tabular}
\end{center}
Assume that $\psi$ restricted to $\ISbr(M)$ is the identity function on 
$\ISbr(M)$.
Then, for each $X \in \ISbr(\Methbr)$, $\psi'(X) \feqv X$ and 
$\len(\psi'(X)) \leq \len(X) + (k - 1) \mul p$, where $p$ is the number 
of occurrences of primitive instructions from 
$\PIbr(\Methbr) \diff \PIbr(M)$ in $X$.
\end{proposition}
\begin{proof}
It is easily proved by induction on the length of $X$ that, for each 
$l,n \in \Nat$, 
$\fjmp{l} \conc \psi'(X) \conc \halt^n \feqv
 \fjmp{l} \conc X \conc \halt^n$.
From this, it follows immediately that $\psi'(X) \feqv X$.

Suppose that $X = u_1 \conc \ldots \conc u_n$.
Let $p$ be the number of occurrences of primitive instructions from 
$\PIbr(\Methbr) \diff \PIbr(M)$ in $X$.
Then 
\begin{flushleft}
\lststretch
\begin{tabular}[t]{@{\hsp{5}}l@{}}
$\sum_{i \in \set{1,\ldots,n}
       \mathrm{\,s.t.\,} u_i \notin \PIbr(\Methbr) \diff \PIbr(M)}
  \len(u_i) \phantom{\psi()} = 
 \len(X) - p\;$;
\\
$\sum_{i \in \set{1,\ldots,n} 
       \mathrm{\,s.t.\,} u_i \in \PIbr(\Methbr) \diff \PIbr(M)}
  \len(\psi(u_i))  \leq
 k \mul p\;$.
\end{tabular}
\end{flushleft}
Hence,
$\len(\psi'(X)) \leq \len(X) - p + k \mul p = \len(X) + (k - 1) \mul p$.
\qed
\end{proof}

We have the following corollary of part~(1) of 
Theorem~\ref{theorem-min-meth-set} and the definition of 
$k$-size-bounded functional completeness.
\begin{corollary}
\label{corollary-bounded-complete-1}
$\PIbr(\set{\mbr{\FFunc}{\FFunc},\mbr{\TFunc}{\TFunc},
            \mbr{\IFunc}{\IFunc},\mbr{\CFunc}{\CFunc},
            \mbr{\IFunc}{\FFunc},\mbr{\IFunc}{\TFunc},
            \mbr{\TFunc}{\IFunc},\mbr{\TFunc}{\CFunc}})$ 
is $1$-size-bounded \sloppy functionally complete.
\end{corollary}
The following theorem concerns the $k$-size-bounded functional 
completeness of a few subsets of this $1$-size-bounded functionally 
complete instruction set, including the ones that we used 
in~\cite{BM13a,BM14e}.
\begin{theorem}
\label{theorem-bounded-complete}
\mbox{} 
\begin{flushleft}
\lststretch
\begin{tabular}[t]{@{}l@{\hsp{.3}}l@{\hsp{.15}}c@{\hsp{.2}}r@{}}
\textup{(1)} &
$\PIbr(\set{\mbr{\FFunc}{\FFunc},\mbr{\TFunc}{\TFunc},
            \mbr{\IFunc}{\IFunc},\mbr{\CFunc}{\CFunc},
            \mbr{\IFunc}{\FFunc},\mbr{\IFunc}{\TFunc}})$
& is & strictly $2$-size-bounded funct.\ compl.
\\
\textup{(2)} &
$\PIbr(\set{\mbr{\FFunc}{\FFunc},\mbr{\TFunc}{\TFunc},
            \mbr{\IFunc}{\IFunc},\mbr{\CFunc}{\CFunc}})$
& is & strictly $3$-size-bounded funct.\ compl.
\\
\textup{(3)} &
$\PIbr(\set{\mbr{\FFunc}{\FFunc},\mbr{\TFunc}{\TFunc},
            \mbr{\IFunc}{\IFunc}})$
& is & strictly $4$-size-bounded funct.\ compl.
\\
\textup{(4)} &
$\PIbr(\set{\mbr{\CFunc}{\CFunc}})$
& is & strictly $3$-size-bounded funct.\ compl.
\\
\textup{(5)} &
$\PIbr(\set{\mbr{\IFunc}{\FFunc},\mbr{\IFunc}{\TFunc}})$
& is & strictly $4$-size-bounded funct.\ compl.
\end{tabular}
\end{flushleft}
\end{theorem}
\begin{proof}
\sloppy
We assume that, for each $M \subseteq \Methbr$ and 
$k \in \Natpos$, the restriction to $\PIbr(M)$ of a function $\psi$ that 
witnesses $k$-size-bounded functional completeness of $\PIbr(M)$ is the
identify function on $\PIbr(M)$ and the restriction to 
$\PIbr(\Methbr) \diff \PIbr(M)$ has the same instruction sequence from 
$\ISbr(M)$ as value for primitive instruction from the same equivalence 
class of $\PIbr(\Methbr)$. 
It is clear that this assumption can be made without loss of generality.

Below, for each individual part of the theorem, first a function $\psi$ 
that witnesses the stated size-bounded functional completeness is 
uniquely characterized by giving the instruction sequences for the 
primitive instructions for Boolean registers that are not covered by the 
assumption and then the strictness of the stated size-bounded functional 
completeness is established by checking for one of the primitive 
instructions concerned for which the given instruction sequence was of 
the greatest length that the given instruction sequence cannot be 
replaced by a shorter one.     

We say that a $u \in \PIbr(\Methbr)$ cannot be replaced by a jump 
instruction if there exists a $b \in \Bool$ and $n \in \Natpos$ such 
that, for each $v \in \set{\fjmp{l} \where l \in \Nat}$,
$\extr{u \conc \halt^n} \sfause f.\BR^\Methbr_b \neq
 \extr{v \conc \halt^n} \sfause f.\BR^\Methbr_b$.
\begin{itemize}
\item[(1)]
Let 
$M = \set{\mbr{\FFunc}{\FFunc},\mbr{\TFunc}{\TFunc},
          \mbr{\IFunc}{\IFunc},\mbr{\CFunc}{\CFunc},
          \mbr{\IFunc}{\FFunc},\mbr{\IFunc}{\TFunc}}$.
Then instruction sequences from $\ISbr(M)$ are needed for 
$\ntst{f.\mbr{\TFunc}{\IFunc}}$ and $\ntst{f.\mbr{\TFunc}{\CFunc}}$.
Take $\psi$ such that
\begin{flushleft}
\lststretch
\begin{tabular}[t]{@{\hsp{5}}l@{\hsp{.7}}l@{}}
\textup{(a)} &
$\psi(\ntst{f.\mbr{\TFunc}{\IFunc}}) = \fjmp{2}$, 
\\
\textup{(b)} &
$\psi(\ntst{f.\mbr{\TFunc}{\CFunc}}) =
 f.\mbr{\CFunc}{\CFunc} \conc \fjmp{2}$.
\end{tabular}
\end{flushleft}
Then $\psi$ witnesses the $2$-size-bounded functional completeness of 
$\PIbr(M)$.
Because $\ntst{f.\mbr{\TFunc}{\CFunc}}$ cannot be replaced by a jump
instruction and there exists no $u \in \PIbr(M)$ such that 
$u \eeqv \ntst{f.\mbr{\TFunc}{\CFunc}}$, $\PIbr(M)$ is not 
$1$-size-bounded functionally complete.
Hence, $\PIbr(M)$ is strictly $2$-size-bounded functionally complete.
\item[(2)]
Let 
$M = \set{\mbr{\FFunc}{\FFunc},\mbr{\TFunc}{\TFunc},
          \mbr{\IFunc}{\IFunc},\mbr{\CFunc}{\CFunc}}$.
Then instruction sequences from $\ISbr(M)$ are needed for 
$\ntst{f.\mbr{\TFunc}{\IFunc}}$, $\ntst{f.\mbr{\TFunc}{\CFunc}}$,
$\ptst{f.\mbr{\IFunc}{\FFunc}}$, $\ntst{f.\mbr{\IFunc}{\FFunc}}$,
$\ptst{f.\mbr{\IFunc}{\TFunc}}$, and $\ntst{f.\mbr{\IFunc}{\TFunc}}$.
Take $\psi$ such that (a), (b), and
\begin{flushleft}
\lststretch
\begin{tabular}[t]{@{\hsp{5}}l@{\hsp{.7}}l@{}}
\textup{(c1)} &
$\psi(\ptst{f.\mbr{\IFunc}{\FFunc}}) =
 \ptst{f.\mbr{\IFunc}{\IFunc}} \conc \ptst{f.\mbr{\FFunc}{\FFunc}} \conc
 \ptst{f.\mbr{\FFunc}{\FFunc}}$ or \\  
\textup{(c2)} &
$\psi(\ptst{f.\mbr{\IFunc}{\FFunc}}) =
 \ntst{f.\mbr{\CFunc}{\CFunc}} \conc \ptst{f.\mbr{\FFunc}{\FFunc}} \conc
 \ptst{f.\mbr{\FFunc}{\FFunc}}$ or \\  
\textup{(c3)} &
$\psi(\ptst{f.\mbr{\IFunc}{\FFunc}}) =
 \ptst{f.\mbr{\IFunc}{\IFunc}} \conc \ptst{f.\mbr{\CFunc}{\CFunc}} \conc
 \fjmp{2}$ or \\  
\textup{(c4)} &
$\psi(\ptst{f.\mbr{\IFunc}{\FFunc}}) =
 \ntst{f.\mbr{\CFunc}{\CFunc}} \conc \fjmp{2} \conc
 \ptst{f.\mbr{\CFunc}{\CFunc}}$,
\\
\textup{(d1)} &
$\psi(\ntst{f.\mbr{\IFunc}{\FFunc}}) =
 \ntst{f.\mbr{\IFunc}{\IFunc}} \conc \ptst{f.\mbr{\FFunc}{\FFunc}} \conc
 \ptst{f.\mbr{\FFunc}{\FFunc}}$ or \\  
\textup{(d2)} &
$\psi(\ntst{f.\mbr{\IFunc}{\FFunc}}) =
 \ptst{f.\mbr{\CFunc}{\CFunc}} \conc \ptst{f.\mbr{\FFunc}{\FFunc}} \conc
 \ptst{f.\mbr{\FFunc}{\FFunc}}$ or \\  
\textup{(d3)} &
$\psi(\ntst{f.\mbr{\IFunc}{\FFunc}}) =
 \ntst{f.\mbr{\IFunc}{\IFunc}} \conc \fjmp{2} \conc
 \ptst{f.\mbr{\CFunc}{\CFunc}}$ or \\  
\textup{(d4)} &
$\psi(\ntst{f.\mbr{\IFunc}{\FFunc}}) =
 \ptst{f.\mbr{\CFunc}{\CFunc}} \conc \ptst{f.\mbr{\CFunc}{\CFunc}} \conc
 \fjmp{2}$,
\\
\textup{(e1)} &
$\psi(\ptst{f.\mbr{\IFunc}{\TFunc}}) =
 \ptst{f.\mbr{\IFunc}{\IFunc}} \conc \ntst{f.\mbr{\TFunc}{\TFunc}} \conc
 \ntst{f.\mbr{\TFunc}{\TFunc}}$ or \\  
\textup{(e2)} &
$\psi(\ptst{f.\mbr{\IFunc}{\TFunc}}) =
 \ntst{f.\mbr{\CFunc}{\CFunc}} \conc \ntst{f.\mbr{\TFunc}{\TFunc}} \conc
 \ntst{f.\mbr{\TFunc}{\TFunc}}$ or \\  
\textup{(e3)} &
$\psi(\ptst{f.\mbr{\IFunc}{\TFunc}}) =
 \ptst{f.\mbr{\IFunc}{\IFunc}} \conc \fjmp{2} \conc
 \ntst{f.\mbr{\CFunc}{\CFunc}}$ or \\  
\textup{(e4)} &
$\psi(\ptst{f.\mbr{\IFunc}{\TFunc}}) =
 \ntst{f.\mbr{\CFunc}{\CFunc}} \conc \ntst{f.\mbr{\CFunc}{\CFunc}} \conc
 \fjmp{2}$,
\end{tabular}
\end{flushleft}
\begin{flushleft}
\lststretch
\begin{tabular}[t]{@{\hsp{5}}l@{\hsp{.7}}l@{}}
\textup{(f1)} &
$\psi(\ntst{f.\mbr{\IFunc}{\TFunc}}) =
 \ntst{f.\mbr{\IFunc}{\IFunc}} \conc \ntst{f.\mbr{\TFunc}{\TFunc}} \conc
 \ntst{f.\mbr{\TFunc}{\TFunc}}$ or \\  
\textup{(f2)} &
$\psi(\ntst{f.\mbr{\IFunc}{\TFunc}}) =
 \ptst{f.\mbr{\CFunc}{\CFunc}} \conc \ntst{f.\mbr{\TFunc}{\TFunc}} \conc
 \ntst{f.\mbr{\TFunc}{\TFunc}}$ or \\  
\textup{(f3)} &
$\psi(\ntst{f.\mbr{\IFunc}{\TFunc}}) =
 \ntst{f.\mbr{\IFunc}{\IFunc}} \conc \ntst{f.\mbr{\CFunc}{\CFunc}} \conc
 \fjmp{2}$ or \\  
\textup{(f4)} &
$\psi(\ntst{f.\mbr{\IFunc}{\TFunc}}) =
 \ptst{f.\mbr{\CFunc}{\CFunc}} \conc \fjmp{2} \conc
 \ntst{f.\mbr{\CFunc}{\CFunc}}$.\footnotemark
\end{tabular}
\footnotetext
{For several instruction sequences that start with a test instruction, 
 there is a counterpart with the same methods in the same numbers that 
 starts with the opposite test instruction.
 We refrain from mentioning these counterparts as alternatives.}
\end{flushleft}
Then $\psi$ witnesses the $3$-size-bounded functional completeness of 
$\PIbr(M)$. 
To obtain the effects of $\ptst{f.\mbr{\IFunc}{\FFunc}}$, an instruction 
sequence from $\ISbr(M)$ is needed that contains a test instruction from 
$\PIbr(M)$ with $\mbr{\IFunc}{\IFunc}$ or $\mbr{\CFunc}{\CFunc}$ as 
method and a primitive instruction from $\PIbr(M)$ with 
$\mbr{\FFunc}{\FFunc}$ or $\mbr{\CFunc}{\CFunc}$ as method.
Because there does not exist such an instruction sequence of length $2$
with the right effects, $\PIbr(M)$ is not $2$-size-bounded functionally 
complete.
Hence, $\PIbr(M)$ is strictly $3$-size-bounded functionally complete.
\item[(3)]
Let 
$M = \set{\mbr{\FFunc}{\FFunc},\mbr{\TFunc}{\TFunc},
          \mbr{\IFunc}{\IFunc}}$.
Then instruction sequences from $\ISbr(M)$ are needed for 
$\ntst{f.\mbr{\TFunc}{\IFunc}}$, $\ptst{f.\mbr{\IFunc}{\FFunc}}$, 
$\ntst{f.\mbr{\IFunc}{\FFunc}}$, $\ptst{f.\mbr{\IFunc}{\TFunc}}$, 
$\ntst{f.\mbr{\IFunc}{\TFunc}}$, $f.\mbr{\CFunc}{\CFunc}$, 
$\ptst{f.\mbr{\CFunc}{\CFunc}}$, $\ntst{f.\mbr{\CFunc}{\CFunc}}$, 
and $\ntst{f.\mbr{\TFunc}{\CFunc}}$.
Take $\psi$ such that (a), (c1) or (c3), (d1) or (d3), (e1) or (e3), 
(f1) or (f3), and
\begin{flushleft}
\lststretch
\begin{tabular}[t]{@{\hsp{5}}l@{\hsp{.7}}l@{}}
\textup{(g)} &
$\psi(f.\mbr{\CFunc}{\CFunc}) \phantom{+} =
 \ptst{f.\mbr{\IFunc}{\IFunc}} \conc \ptst{f.\mbr{\FFunc}{\FFunc}} \conc
 f.\mbr{\TFunc}{\TFunc}$, 
\\  
\textup{(h)} &
$\psi(\ptst{f.\mbr{\CFunc}{\CFunc}}) =
 \ntst{f.\mbr{\IFunc}{\IFunc}} \conc \ntst{f.\mbr{\TFunc}{\TFunc}} \conc
 \ptst{f.\mbr{\FFunc}{\FFunc}}$, 
\\  
\textup{(i)} &
$\psi(\ntst{f.\mbr{\CFunc}{\CFunc}}) =
 \ptst{f.\mbr{\IFunc}{\IFunc}} \conc \ptst{f.\mbr{\FFunc}{\FFunc}} \conc
 \ntst{f.\mbr{\TFunc}{\TFunc}}$,
\\
\textup{(j)} &
$\psi(\ntst{f.\mbr{\TFunc}{\CFunc}}) =
 \ptst{f.\mbr{\IFunc}{\IFunc}} \conc \ptst{f.\mbr{\FFunc}{\FFunc}} \conc
 f.\mbr{\TFunc}{\TFunc} \conc \fjmp{2}$.
\end{tabular}
\end{flushleft}
Then $\psi$ witnesses the $4$-size-bounded functional completeness of 
$\PIbr(M)$.
To obtain the effects of $\ntst{f.\mbr{\TFunc}{\CFunc}}$, an instruction 
sequence from $\ISbr(M)$ is needed that contains a test instruction from
$\PIbr(M)$ with $\mbr{\IFunc}{\IFunc}$ as method, a primitive 
instruction from $\PIbr(M)$ with $\mbr{\FFunc}{\FFunc}$ as method, and a 
primitive instruction from $\PIbr(M)$ with $\mbr{\TFunc}{\TFunc}$ as 
method.
Because there does not exist such an instruction sequence of length $3$
with the right effects, $\PIbr(M)$ is not $3$-size-bounded functionally 
complete.
Hence, $\PIbr(M)$ is strictly $4$-size-bounded functionally complete.
\item[(4)]
Let 
$M = \set{\mbr{\CFunc}{\CFunc}}$.
Then instruction sequences from $\ISbr(M)$ are needed for 
$\ntst{f.\mbr{\TFunc}{\IFunc}}$, $\ntst{f.\mbr{\TFunc}{\CFunc}}$,
$\ptst{f.\mbr{\IFunc}{\FFunc}}$, $\ntst{f.\mbr{\IFunc}{\FFunc}}$,
$\ptst{f.\mbr{\IFunc}{\TFunc}}$, $\ntst{f.\mbr{\IFunc}{\TFunc}}$,
$f.\mbr{\FFunc}{\FFunc}$, $\ptst{f.\mbr{\FFunc}{\FFunc}}$, 
$f.\mbr{\TFunc}{\TFunc}$, $\ntst{f.\mbr{\TFunc}{\TFunc}}$,
$f.\mbr{\IFunc}{\IFunc}$, $\ptst{f.\mbr{\IFunc}{\IFunc}}$, and
$\ntst{f.\mbr{\IFunc}{\IFunc}}$. 
Take $\psi$ such that (a), (b), (c4), (d4), (e4), (f4), and
\begin{flushleft}
\lststretch
\begin{tabular}[t]{@{\hsp{5}}l@{\hsp{.7}}l@{}}
\textup{(k)} &
$\psi(f.\mbr{\FFunc}{\FFunc}) \phantom{+} =
 \ptst{f.\mbr{\CFunc}{\CFunc}} \conc f.\mbr{\CFunc}{\CFunc}$, 
\\  
\textup{(l)} &
$\psi(\ptst{f.\mbr{\FFunc}{\FFunc}}) =
 \ptst{f.\mbr{\CFunc}{\CFunc}} \conc f.\mbr{\CFunc}{\CFunc} \conc
 \fjmp{2}$, 
\\
\textup{(m)} &
$\psi(f.\mbr{\TFunc}{\TFunc}) \phantom{+} =
 \ntst{f.\mbr{\CFunc}{\CFunc}} \conc f.\mbr{\CFunc}{\CFunc}$, 
\\
\textup{(n)} &
$\psi(\ntst{f.\mbr{\TFunc}{\TFunc}}) =
 \ntst{f.\mbr{\CFunc}{\CFunc}} \conc f.\mbr{\CFunc}{\CFunc} \conc
 \fjmp{2}$,
\end{tabular}
\end{flushleft}
\begin{flushleft}
\lststretch
\begin{tabular}[t]{@{\hsp{5}}l@{\hsp{.7}}l@{}}
\textup{(o)} &
$\psi(f.\mbr{\IFunc}{\IFunc}) \phantom{+} =
 f.\mbr{\CFunc}{\CFunc} \conc f.\mbr{\CFunc}{\CFunc}$, 
\\  
\textup{(p)} &
$\psi(\ptst{f.\mbr{\IFunc}{\IFunc}}) =
 f.\mbr{\CFunc}{\CFunc} \conc \ptst{f.\mbr{\CFunc}{\CFunc}}$, 
\\  
\textup{(q)} &
$\psi(\ntst{f.\mbr{\IFunc}{\IFunc}}) =
 f.\mbr{\CFunc}{\CFunc} \conc \ntst{f.\mbr{\CFunc}{\CFunc}}$. 
\end{tabular}
\end{flushleft}
Then $\psi$ witnesses the $3$-size-bounded functional completeness of 
$\PIbr(M)$.
To obtain the effects of $\ptst{f.\mbr{\FFunc}{\FFunc}}$, an instruction 
sequence from $\ISbr(M)$ is needed.
Because there does not exist such an instruction sequence of length $2$ 
with the right effects, $\PIbr(M)$ is not $2$-size-bounded functionally 
complete.
Hence, $\PIbr(M)$ is strictly $3$-size-bounded functionally complete.
\item[(5)]
Let 
$M = \set{\mbr{\IFunc}{\FFunc},\mbr{\IFunc}{\TFunc}}$.
Then instruction sequences from $\ISbr(M)$ are needed for 
$\ntst{f.\mbr{\TFunc}{\IFunc}}$, $f.\mbr{\FFunc}{\FFunc}$, 
$\ptst{f.\mbr{\FFunc}{\FFunc}}$, $f.\mbr{\TFunc}{\TFunc}$, 
$\ntst{f.\mbr{\TFunc}{\TFunc}}$, $f.\mbr{\IFunc}{\IFunc}$, 
$\ptst{f.\mbr{\IFunc}{\IFunc}}$, $\ntst{f.\mbr{\IFunc}{\IFunc}}$, 
$f.\mbr{\CFunc}{\CFunc}$, $\ptst{f.\mbr{\CFunc}{\CFunc}}$, 
$\ntst{f.\mbr{\CFunc}{\CFunc}}$, and $\ntst{f.\mbr{\TFunc}{\CFunc}}$.
Take $\psi$ such that (a) and
\begin{flushleft}
\lststretch
\begin{tabular}[t]{@{\hsp{5}}l@{\hsp{.7}}l@{}}
\textup{(r)} &
$\psi(f.\mbr{\FFunc}{\FFunc}) \phantom{+} =
 f.\mbr{\IFunc}{\FFunc}$, 
\\  
\textup{(s)} &
$\psi(\ptst{f.\mbr{\FFunc}{\FFunc}}) =
 f.\mbr{\IFunc}{\FFunc} \conc \fjmp{2}$, 
\\  
\textup{(t)} &
$\psi(f.\mbr{\TFunc}{\TFunc}) \phantom{+} =
 f.\mbr{\IFunc}{\TFunc}$, 
\\
\textup{(u)} &
$\psi(\ntst{f.\mbr{\TFunc}{\TFunc}}) =
 f.\mbr{\IFunc}{\TFunc} \conc \fjmp{2}$, 
\\
\textup{(v1)} &
$\psi(f.\mbr{\IFunc}{\IFunc}) \phantom{+} =
 \ptst{f.\mbr{\IFunc}{\FFunc}} \conc \ptst{f.\mbr{\IFunc}{\TFunc}} \conc
 \ntst{f.\mbr{\IFunc}{\FFunc}}$ or
\\  
\textup{(v2)} &
$\psi(f.\mbr{\IFunc}{\IFunc}) \phantom{+} =
 \ptst{f.\mbr{\IFunc}{\TFunc}} \conc \ntst{f.\mbr{\IFunc}{\TFunc}} \conc
 \ptst{f.\mbr{\IFunc}{\FFunc}}$ or
\\  
\textup{(v3)} &
$\psi(f.\mbr{\IFunc}{\IFunc}) \phantom{+} =
 \ptst{f.\mbr{\IFunc}{\FFunc}} \conc \ptst{f.\mbr{\IFunc}{\TFunc}} \conc
 \fjmp{1}$ or
\\  
\textup{(v4)} &
$\psi(f.\mbr{\IFunc}{\IFunc}) \phantom{+} =
 \ptst{f.\mbr{\IFunc}{\TFunc}} \conc \fjmp{2} \conc
 \ptst{f.\mbr{\IFunc}{\FFunc}}$, 
\\
\textup{(w1)} &
$\psi(\ptst{f.\mbr{\IFunc}{\IFunc}}) =
 \ptst{f.\mbr{\IFunc}{\FFunc}} \conc \ptst{f.\mbr{\IFunc}{\TFunc}} \conc
 \ptst{f.\mbr{\IFunc}{\FFunc}}$ or 
\\  
\textup{(w2)} &
$\psi(\ptst{f.\mbr{\IFunc}{\IFunc}}) =
 \ptst{f.\mbr{\IFunc}{\TFunc}} \conc \ntst{f.\mbr{\IFunc}{\TFunc}} \conc
 \ntst{f.\mbr{\IFunc}{\FFunc}}$ or 
\\  
\textup{(w3)} &
$\psi(\ptst{f.\mbr{\IFunc}{\IFunc}}) =
 \ptst{f.\mbr{\IFunc}{\FFunc}} \conc \ptst{f.\mbr{\IFunc}{\TFunc}} \conc
 \fjmp{2}$ or 
\\  
\textup{(w4)} &
$\psi(\ptst{f.\mbr{\IFunc}{\IFunc}}) =
 \ptst{f.\mbr{\IFunc}{\TFunc}} \conc \fjmp{2} \conc
 \ntst{f.\mbr{\IFunc}{\FFunc}}$, 
\\  
\textup{(x1)} &
$\psi(\ntst{f.\mbr{\IFunc}{\IFunc}}) =
 \ntst{f.\mbr{\IFunc}{\FFunc}} \conc \ptst{f.\mbr{\IFunc}{\FFunc}} \conc
 \ptst{f.\mbr{\IFunc}{\TFunc}}$ or
\\
\textup{(x2)} &
$\psi(\ntst{f.\mbr{\IFunc}{\IFunc}}) =
 \ntst{f.\mbr{\IFunc}{\TFunc}} \conc \ntst{f.\mbr{\IFunc}{\FFunc}} \conc
 \ntst{f.\mbr{\IFunc}{\TFunc}}$ or
\\
\textup{(x3)} &
$\psi(\ntst{f.\mbr{\IFunc}{\IFunc}}) =
 \ntst{f.\mbr{\IFunc}{\FFunc}} \conc \fjmp{2} \conc
 \ptst{f.\mbr{\IFunc}{\TFunc}}$ or
\\
\textup{(x4)} &
$\psi(\ntst{f.\mbr{\IFunc}{\IFunc}}) =
 \ntst{f.\mbr{\IFunc}{\TFunc}} \conc \ntst{f.\mbr{\IFunc}{\FFunc}} \conc
 \fjmp{2}$,
\\
\textup{(y1)} &
$\psi(f.\mbr{\CFunc}{\CFunc}) \phantom{+} =
 \ptst{f.\mbr{\IFunc}{\FFunc}} \conc \ptst{f.\mbr{\IFunc}{\FFunc}} \conc
 \ntst{f.\mbr{\IFunc}{\TFunc}}$ or
\\  
\textup{(y2)} &
$\psi(f.\mbr{\CFunc}{\CFunc}) \phantom{+} =
 \ptst{f.\mbr{\IFunc}{\TFunc}} \conc \ntst{f.\mbr{\IFunc}{\FFunc}} \conc
 \ptst{f.\mbr{\IFunc}{\TFunc}}$ or 
\\
\textup{(y3)} &
$\psi(f.\mbr{\CFunc}{\CFunc}) \phantom{+} =
 \ptst{f.\mbr{\IFunc}{\FFunc}} \conc \fjmp{2} \conc
 \ntst{f.\mbr{\IFunc}{\TFunc}}$ or
\\  
\textup{(y4)} &
$\psi(f.\mbr{\CFunc}{\CFunc}) \phantom{+} =
 \ptst{f.\mbr{\IFunc}{\TFunc}} \conc \ntst{f.\mbr{\IFunc}{\FFunc}} \conc
 \fjmp{1}$, 
\\  
\textup{(z1)} &
$\psi(\ptst{f.\mbr{\CFunc}{\CFunc}}) =
 \ntst{f.\mbr{\IFunc}{\FFunc}} \conc \ptst{f.\mbr{\IFunc}{\TFunc}} \conc
 \ptst{f.\mbr{\IFunc}{\FFunc}}$ or 
\\  
\textup{(z2)} &
$\psi(\ptst{f.\mbr{\CFunc}{\CFunc}}) =
 \ntst{f.\mbr{\IFunc}{\TFunc}} \conc \ntst{f.\mbr{\IFunc}{\TFunc}} \conc
 \ntst{f.\mbr{\IFunc}{\FFunc}}$ or 
\\  
\textup{(z3)} &
$\psi(\ptst{f.\mbr{\CFunc}{\CFunc}}) =
 \ntst{f.\mbr{\IFunc}{\FFunc}} \conc \ptst{f.\mbr{\IFunc}{\TFunc}} \conc
 \fjmp{2}$ or 
\\  
\textup{(z4)} &
$\psi(\ptst{f.\mbr{\CFunc}{\CFunc}}) =
 \ntst{f.\mbr{\IFunc}{\TFunc}} \conc \fjmp{2} \conc
 \ntst{f.\mbr{\IFunc}{\FFunc}}$, 
\end{tabular}
\end{flushleft}
\begin{flushleft}
\lststretch
\begin{tabular}[t]{@{\hsp{5}}l@{\hsp{.7}}l@{}}
\textup{(aa1)} &
$\psi(\ntst{f.\mbr{\CFunc}{\CFunc}}) =
 \ptst{f.\mbr{\IFunc}{\FFunc}} \conc \ptst{f.\mbr{\IFunc}{\FFunc}} \conc
 \ptst{f.\mbr{\IFunc}{\TFunc}}$ or
\\
\textup{(aa2)} &
$\psi(\ntst{f.\mbr{\CFunc}{\CFunc}}) =
 \ptst{f.\mbr{\IFunc}{\TFunc}} \conc \ntst{f.\mbr{\IFunc}{\FFunc}} \conc
 \ntst{f.\mbr{\IFunc}{\TFunc}}$ or
\\
\textup{(aa3)} &
$\psi(\ntst{f.\mbr{\CFunc}{\CFunc}}) =
 \ptst{f.\mbr{\IFunc}{\FFunc}} \conc \fjmp{2} \conc
 \ptst{f.\mbr{\IFunc}{\TFunc}}$ or
\\
\textup{(aa4)} &
$\psi(\ntst{f.\mbr{\CFunc}{\CFunc}}) =
 \ptst{f.\mbr{\IFunc}{\TFunc}} \conc \ntst{f.\mbr{\IFunc}{\FFunc}} \conc
 \fjmp{2}$,
\\
\textup{(ab1)} &
$\psi(\ntst{f.\mbr{\TFunc}{\CFunc}}) =
 \ptst{f.\mbr{\IFunc}{\FFunc}} \conc \ptst{f.\mbr{\IFunc}{\FFunc}} \conc
 \ntst{f.\mbr{\IFunc}{\TFunc}} \conc \fjmp{2}$ or
\\
\textup{(ab2)} &
$\psi(\ntst{f.\mbr{\TFunc}{\CFunc}}) =
 \ptst{f.\mbr{\IFunc}{\TFunc}} \conc \ntst{f.\mbr{\IFunc}{\FFunc}} \conc
 \ptst{f.\mbr{\IFunc}{\TFunc}} \conc \fjmp{2}$ or
\\
\textup{(ab3)} &
$\psi(\ntst{f.\mbr{\TFunc}{\CFunc}}) =
 \ptst{f.\mbr{\IFunc}{\FFunc}} \conc \fjmp{2} \conc
 \ntst{f.\mbr{\IFunc}{\TFunc}} \conc \fjmp{2}$ or
\\
\textup{(ab4)} &
$\psi(\ntst{f.\mbr{\TFunc}{\CFunc}}) =
 \ptst{f.\mbr{\IFunc}{\TFunc}} \conc \ntst{f.\mbr{\IFunc}{\FFunc}} \conc
 \fjmp{1} \conc \fjmp{2}$.
\end{tabular}
\end{flushleft}
Then $\psi$ witnesses the $4$-size-bounded functional completeness of 
$\PIbr(M)$.
To obtain the effects of $\ntst{f.\mbr{\TFunc}{\CFunc}}$, an instruction 
sequence from $\ISbr(M)$ is needed.
Because there does not exist such an instruction sequence of length $3$
with the right effects, $\PIbr(M)$ is not $3$-size-bounded functionally 
complete.
Hence, $\PIbr(M)$ is strictly $4$-size-bounded functionally complete. 
\qed
\end{itemize}
\end{proof}
Theorem~\ref{theorem-bounded-complete} tells us among other things 
that the instruction sets 
$\PIbr(\set{\mbr{\CFunc}{\CFunc}})$ 
and 
$\PIbr(\set{\mbr{\FFunc}{\FFunc},\mbr{\TFunc}{\TFunc},
            \mbr{\IFunc}{\IFunc},\mbr{\CFunc}{\CFunc}})$
are both strictly $3$-size-bounded functionally complete.
However, the latter instruction set often gives rise to shorter 
instruction sequences than the former instruction set because the 
effects of a primitive instruction from the set
$\PIbr(\set{\mbr{\FFunc}{\FFunc},\mbr{\TFunc}{\TFunc},
            \mbr{\IFunc}{\IFunc}})$
do not have to be obtained by means of two or three primitive 
instructions from the set 
$\PIbr(\set{\mbr{\CFunc}{\CFunc}})$.

We have the following corollary of the proof of 
Theorem~\ref{theorem-bounded-complete}.
\begin{corollary}
\label{corollary-bounded-complete-2}
Let $M \subset \set{\mbr{\FFunc}{\FFunc},\mbr{\TFunc}{\TFunc},
                    \mbr{\IFunc}{\IFunc},\mbr{\CFunc}{\CFunc},
                    \mbr{\IFunc}{\FFunc},\mbr{\IFunc}{\TFunc}}$.
Then:
\begin{itemize}
\item[\textup{(1)}]
$\PIbr(M)$ is strictly $4$-size-bounded functionally complete
if $\mbr{\CFunc}{\CFunc} \notin M$ and 
either
$\set{\mbr{\FFunc}{\FFunc},\mbr{\TFunc}{\TFunc},\mbr{\IFunc}{\IFunc}}
  \subseteq M$
or
$\set{\mbr{\IFunc}{\FFunc},\mbr{\IFunc}{\TFunc}}  \subseteq M$;
\item[\textup{(2)}]
$\PIbr(M)$ is strictly $3$-size-bounded functionally complete
if $\mbr{\CFunc}{\CFunc} \in M$.
\end{itemize}
\end{corollary}

\section{Concluding Remarks}
\label{sect-concl}

We have investigated instruction sequence size bounded functional 
completeness of instruction sets for Boolean registers. 
Our main results are Corollaries~\ref{corollary-bounded-complete-1}
and~\ref{corollary-bounded-complete-2}.
The latter corollary covers 44 instruction sets.
The covered instruction sets include the instruction sets that we used 
earlier in~\cite{BM13a,BM14e} and many other relatively obvious 
instruction sets.
The covered instruction sets belong to the 255 instruction sets that are 
non-empty subsets of one of the 256 instruction sets with the 
property that each possible instruction has the same effects as one 
from the set (see Corollary~\ref{corollary-min-meth-set}).
It is still an open question what is the smallest $k$ such that each of 
these 255 instruction sets is $k$-size-bounded functionally complete if 
it is $k'$-size-bounded functionally complete for some $k'$.

In our work on instruction sequence size complexity presented 
in~\cite{BM13a}, we have established several connections between 
instruction sequence based complexity theory and classical complexity 
theory.
For example, we have introduced instruction sequence based counterparts 
of the complexity classes P/poly and NP/poly and we have formulated an 
instruction sequence based counterpart of the well-known 
complexity-theoretic conjecture that \mbox{NP $\not\subseteq$ P/poly}.%
\footnote
{The non-uniform complexity classes P/poly and NP/poly, as well as the
 conjecture that NP $\not\subseteq$ P/poly, are treated in many
 textbooks on classical complexity theory 
 (see e.g.~\cite{AB09a,HS11a,Weg05a}).}
However, for many a question that arises naturally with the approach to 
complexity based on instruction sequence size, it is far from obvious 
whether a comparable question can be raised in classical complexity 
theory based on Turing machines or Boolean circuits.
In particular, this is far from obvious for questions concerning 
instruction sets for Boolean registers.
 
\subsection*{Acknowledgement}

We thank three anonymous referees for their helpful suggestions.

\bibliographystyle{splncs03}
\bibliography{IS}

\end{document}